\documentclass[lettersize,journal]{IEEEtran}

\usepackage{cite}
\usepackage{multicol}
\usepackage{cuted}
\usepackage{lipsum}
\usepackage{mathtools}
\usepackage{amsmath,amssymb,amsfonts}
\usepackage{mathtools, nccmath}
\usepackage{algorithmic}
\usepackage{graphicx}
\usepackage{textcomp}
\usepackage{xcolor}
\newtheorem{theorem}{Theorem}
\newtheorem{lemma}{Lemma}
\newtheorem{remark}{Remark}
\newtheorem{corollary}{Corollary}
\newtheorem{proof}{Proof}

\def\BibTeX{{\rm B\kern-.05em{\sc i\kern-.025em b}\kern-.08em
		T\kern-.1667em\lower.7ex\hbox{E}\kern-.125emX}}

\usepackage{stfloats}

\usepackage{bbm}

\usepackage{bm}

\usepackage{amsmath,amsfonts}
\usepackage{algorithmic}
\usepackage{array}
\usepackage[caption=false,font=normalsize,labelfont=sf,textfont=sf]{subfig}
\usepackage{textcomp}
\usepackage{stfloats}
\usepackage{url}
\usepackage{verbatim}
\usepackage{graphicx}
\def\BibTeX{{\rm B\kern-.05em{\sc i\kern-.025em b}\kern-.08em
  T\kern-.1667em\lower.7ex\hbox{E}\kern-.125emX}}
\usepackage{balance}

\begin{document}

\setlength{\textfloatsep}{0.11cm}
 \setlength{\abovedisplayskip}{0.10cm}
 \setlength{\belowdisplayskip}{0.10cm}
 \setlength{\abovecaptionskip}{-1mm}

\title{Capacity Region of Asynchronous Multiple Access Channels with FTN}
\author{Zichao~Zhang,~\IEEEmembership{Student Member,~IEEE,}
		Melda~Yuksel,~\IEEEmembership{Senior Member,~IEEE,}
 Gokhan~M.~Guvensen,~\IEEEmembership{Member,~IEEE,}
		and~Halim~Yanikomeroglu,~\IEEEmembership{Fellow,~IEEE}
\thanks{This work was supported in
part by the Natural Sciences and Engineering Research Council of
Canada, NSERC, under a Discovery Grant and in part by the Scientific and
Technological Research Council of Turkey, TUBITAK, under Grant
122E248.
}\thanks{Z. Zhang and H. Yanikomeroglu are with the Department of Systems and Computer Engineering at Carleton University, Ottawa, ON, K1S 5B6, Canada e-mail: zichaozhang@cmail.carleton.ca,	halim@sce.carleton.ca.}
		\thanks{M. Yuksel and G. M. Guvensen are with the Department of Electrical and Electronics Engineering, Middle East Technical University, Ankara, 06800, Turkey, e-mail: ymelda@metu.edu.tr, guvensen@metu.edu.tr.}}


\maketitle

\begin{abstract}

This paper { derives} the capacity region of asynchronous multiple access channel (MAC) with faster-than-Nyquist (FTN) signaling. We first express the capacity region in the frequency domain. Next, we prove that the capacity definition for finite memory MAC can be generalized to infinite memory MAC. The achievable rate region utilizing the finite memory definition in fact achieves the same region calculated in the frequency domain. Our analysis shows that power optimization is { necessary to achieve the capacity region for asynchronous MAC and FTN}. 

\end{abstract}

\begin{IEEEkeywords}
Capacity, faster-than-Nyquist (FTN), multiple access channel (MAC), asynchronous transmission.
\end{IEEEkeywords}

\section{Introduction}
The rapid growth 
 of need in rate and number of devices proposes a challenge to modern communication systems. Multiple access communications is considered to be one of the potential solutions for 5G and beyond \cite{noma5g}.
Compared to orthogonal multiple access (OMA), multiple access performs non-orthogonal resource allocation. For instance, one frequency band can be shared by more than one user. Besides increased connectivity, multiple access achieves rate pairs that OMA is not able to achieve.  

Faster-than-Nyquist signaling is another promising physical layer technology for future communication systems\cite{mazo}. It improves spectral efficiency by increasing signaling rate, while maintaining the same power consumption\cite{rusek}. 
Since the groundbreaking work of Mazo in 1975 \cite{mazo}, 
there has been a
substantial amount of research on FTN \cite{surveyftn}. The information-theoretical study shows that applying FTN to communication systems improves capacity \cite{rusek} and this improvement becomes more favorable when FTN is applied to multi-antenna communication systems \cite{ourpaper}. 


To support multiple devices sharing the same resources as well as satisfying rate requirements, it is beneficial to exploit the multiple access channel (MAC) with FTN. However, in practice, each device will experience a random time delay. Instead of being a hazard to the system, this asynchronism 
is analyzed in \cite{verdu,ganji} and \cite{outperform} and is shown to be beneficial to multiple access transmission. In \cite{verdu}, the author explored the capacity region of asynchronous MAC with fixed or random time delay differences and showed that these differences bring in additional gains. However, \cite{verdu} is constrained to rectangular pulse shapes in time. The authors of \cite{ganji} removed this limitation and derived the capacity region for band-limited pulse shapes. In \cite{outperform}, the authors studied achievable rates for uplink NOMA with FTN for random link delays and fixed power allocation, and they find that asynchronous transmission is advantageous. In this paper, we derive the capacity region of the asynchronous MAC with FTN with fixed delays, which is significantly larger than \cite{outperform}. 

The organization of the paper is as follows. In Section II we establish the system model. In Section III we derive the capacity region. In Section IV we show that the capacity region for discrete MAC with finite memory defined in \cite{verdu} actually leads to the same region as in Section III. In Section V we plot the rate regions for a finite number of symbols and in Section VI we conclude the paper.

\section{System Model}
The MAC is composed of $K$ transmitters and one receiver. Due to imperfect clock generation or different propagation delays, signals coming from each transmitter have different time delays. We denote them as $\tau_1, \tau_2, \dots, \tau_K$, where $\tau_k\in[0,T], k=1,\dots K$. Without loss of generality, we assume $\tau_1\leq\tau_2\dots\leq\tau_K$. 

All the transmitters use the same pulse shaping filter $p(t)$ and the same acceleration factor $\delta$ for FTN. The signal transmitted from the $k$th user, $x_{k}(t)$ then has the form 
\begin{equation}
	x_{k}(t) = \sum_{m=0}^{N-1}a_{k}[m]p(t-m\delta T-\tau_{k}) \label{eqn:txsignal},
\end{equation}
where $a_{k}[m]$ are the symbols transmitted from the $k$th user and $N$ is the number of symbols transmitted. At the receiver, the matched filter $p^*(-t)$ is applied. 

An additive white Gaussian noise $\xi(t)$ with power spectral density $\sigma_0^2$ is added at the receiver. After passing through the matched filter this white noise becomes correlated. We denote this noise as $\eta(t)=\xi(t)\star p^*(-t)$, where $\star$ denotes the convolution operation. The signal at the output of the matched filter is $y(t)=\left(\sum_{k=1}^{K}x_{k}(t)+\xi(t)\right)\star p^*(-t)$.

In order to obtain sufficient statistics in this asynchronous MAC with FTN, we need to sample according to the time delay of each user. Thus, we sample at all $t=n\delta T+\tau_k, n=0,1,\dots, N-1, k=1, \dots, K$ and obtain $K$ sets of samples instead of a single set \cite{ganji}. 
Then, the samples $y_{k}[n]$ corresponding to user $k$ are written by sampling the output of the matched filter, $y(t)$, at time $n\delta T+\tau_{k}, n=0, \dots, N-1$, and we write
\begin{align}
	y_{ k}[n] 
 &=\sum_{{ l}=1}^{K}\sum_{m=0}^{N-1}a_l[m]g\big((n-m)\delta T+(\tau_{ k}-\tau_{ l})\big) + \eta_{ k}[n]. \label{eqn:modeleq3}
\end{align}
Here $g(t)=p(t)\star p^*(-t)$. Furthermore,
\begin{equation}
	\eta_{k}[n]=\eta(n\delta T +\tau_{k}) = \xi(t)\star p^*(-t)|_{t=n\delta T +\tau_{k}}. \label{eqn:noise}
\end{equation}

By defining the $N\times 1$ vectors $\bm{y}_{k}$, $\bm{a}_{k}$ and $\bm{\eta}_{k}$ to represent respectively the output samples, data symbols and noise, the input-output relationship in \eqref{eqn:modeleq3} can be written in a compact matrix product form as
\begin{equation}
	\left[\begin{matrix}
		\bm{y}_1 \\\bm{y}_2 \\\vdots \\\bm{y}_K
	\end{matrix}\right]= \left[\begin{matrix}
	\bm{G}_{11} & \bm{G}_{12} & \dots &\bm{G}_{1K}\\ \bm{G}_{21} & \bm{G}_{22} & \dots &\bm{G}_{2K} \\ \vdots & \vdots & \ddots & \vdots \\ \bm{G}_{K1} & \bm{G}_{K2} & \dots &\bm{G}_{KK} 
\end{matrix}\right] \left[\begin{matrix} \bm{a}_1 \\ \bm{a}_2 \\ \vdots \\ \bm{a}_K
\end{matrix}\right] + \left[\begin{matrix} \bm{\eta}_1 \\ \bm{\eta}_2 \\ \vdots \\ \bm{\eta}_K
\end{matrix}\right].  \label{eqn:matproduct}
\end{equation}
This expression can further be simplified as
\begin{equation}
	\bm{y}=\tilde{\bm{G}}\bm{a}+\bm{\eta},  \label{eqn:model2ue}
\end{equation}
where $\bm{y}=[\bm{y}_1^\top,\dots,\bm{y}_K^\top]^\top$, $\bm{a}=[\bm{a}_1^\top,\dots,\bm{a}_K^\top]^\top$ and $\bm{\eta}=[\bm{\eta}_1^\top,\dots,\bm{\eta}_K^\top]^\top$. The matrix $\tilde{\mathbf{G}}$ in \eqref{eqn:model2ue} is $KN \times KN$. The matrix $\bm{G}_{kl}$ in \eqref{eqn:matproduct} is the $N\times N$ interference matrix. It represents user ${l}$'s effect on the samples of user ${k}$ and its $(n,m)$th entry, $n,m = 1,\ldots,N$, is 
$(\bm{G}_{kl})_{n,m}=g\big((n-m)\delta T+(\tau_{k}-\tau_{l})\big)$. In this paper, we focus on the special case of $K=2$. Note that, the matrix $\bm{G}_{kl}$ is a Toeplitz matrix. An $N\times N$ Toeplitz matrix $\bm{T}_N$ has the structure $(\bm{T}_N)_{i,j}=t_{i-j}, i,j=0,\dots,N-1$. Its generating function $\mathcal{G}$ is defined as 
\begin{equation}
  \mathcal{G}(\bm{T}_N)=\sum_{k=-\infty}^{\infty}t_ke^{j2\pi \lambda k}, \lambda\in\left[-\frac{1}{2},\frac{1}{2}\right]. \label{eqn:genefuncdef}
\end{equation}


\section{The Capacity Region Analysis}
 In this section we derive the capacity region in the frequency domain. The capacity region $C$ of the two-user multiple access channel with memory is defined as \cite{ganji} 
\begin{equation}
	\begin{aligned}
		C= \underset{{\int_{-\frac{1}{2}}^{\frac{1}{2}}S_k(\lambda)d\lambda\leq P_k,~ S_k(\lambda)\geq 0,~\lambda\in\left[-\frac{1}{2},\frac{1}{2}\right],~k=1,2}}{\bigcup}\bigg\{&(R_1, R_2): \notag 
 \end{aligned}\end{equation}
 \begin{equation}
 \begin{aligned}
 &~~~~~~~~~~~~~~0\leq R_1\leq \underset{N\rightarrow \infty}{\lim}\frac{1}{N}I_N(\bm{a}_1;\bm{y}| \bm{a}_2) \label{eqn:capregiondef}\\
		 &~~~~~~~~~~~~~~0\leq R_2\leq \underset{N\rightarrow \infty}{\lim}\frac{1}{N}I_N(\bm{a}_2;\bm{y}| \bm{a}_1) \\
		 &~~~~~~~~~~~~~~0\leq R_1+R_2\leq \underset{N\rightarrow \infty}{\lim}\frac{1}{N}I_N(\bm{a}_1, \bm{a}_2;\bm{y}) \bigg\},
	\end{aligned}
\end{equation}
 where $S_1(f_n)$ and $S_2(f_n)$ are the power spectral densities of user 1 and user 2, while $P_1$ and $P_2$ are the power constraints. In \eqref{eqn:capregiondef}, $I_N$ is the mutual information between two random vectors with length $N$ \footnote{{ The sum rate in \cite{outperform} is calculated as $I_N(\bm{a}_1; \bm{y}_1|\bm{a}_2)+I_N(\bm{a}_2; \bm{y}_2)$.}}.

 In FTN signaling, the input power spectrum to the physical channel contains the effect of both data symbols as well as FTN \cite{ourpaper,ganji}. This can be written as 
 \begin{equation}
   S_k(\lambda)=\frac{1}{\delta T}G_\delta(\lambda)S_{ak}(\lambda), 
 \end{equation}
where $G_\delta(\lambda)$ is the folded spectrum defined as
 \begin{equation}
   G_\delta(\lambda) = \frac{1}{\delta T}\sum_{n=-\infty}^{\infty}\left|P\left(\frac{\lambda-n}{\delta T}\right)\right|^2=\frac{1}{\delta T}\sum_{n=-\infty}^{\infty}G\left(\frac{\lambda-n}{\delta T}\right)\label{eqn:folded}
 \end{equation}
 and $P(\cdot)$ and $G(\cdot)$ are respectively the continuous time Fourier transforms of $p(t)$ and $g(t)$. The data power spectrum $S_{ak}(\lambda), k=1,2,$ is 
 obtained by the discrete-time Fourier transform of the autocorrelation function of input symbols, $R_{ak}[n]=\mathbb{E}[a_k[m+n]a_k^*[m]]$; i.e., 
 \begin{equation}
  S_{ak}(\lambda)=\sum_{n=-\infty}^{\infty}R_{ak}[n]e^{-j2\pi\lambda n}, k=1,2. \label{eqn:dataspect}
 \end{equation}
 Therefore the power constraint of user $k$ is 
 \begin{equation}
   \frac{1}{\delta T}\int_{-\frac{1}{2}}^{\frac{1}{2}}G_\delta(f_n)S_k(f_n)df_n\leq P_k.
 \end{equation}


In order to obtain a closed-form expression for \eqref{eqn:capregiondef}, we need to calculate the mutual information expressions. The differential entropy of a Gaussian vector $\bm{y}$ is 
\begin{equation}
	h(\bm{y})=\frac{1}{2}\log_2((2\pi)^{2N}\det(\bm{\Sigma_y})), \label{eqn:yentropy}
\end{equation}
where $\bm{\Sigma_y}=\mathbb{E}[\bm{y}\bm{y}^\dagger]$, with $\dagger$ denoting the Hermitian conjugation. Define matrix $\bm{G}=\bm{G}_{11}=\bm{G}_{22}$, the $(n,m)$th entry of which is $g((n-m)\delta T)$, it is easy to see that $\bm{G}$ is a Hermitian matrix. 
Notice that $\bm{G}_{12}^\dagger=\bm{G}_{21}$, thus $\tilde{\bm{G}}$ is a Hermitian matrix. 

The colored Gaussian noise vector $\bm{\eta}$ has the correlation $\mathbb{E}[\eta_i[n]\eta_j[m]]=\sigma_0^2(\bm{G}_{ij})_{n,m}, \quad i, j \in \{1,2\}, n, m\in\{0,1,\dots, N-1\}$. 
As this noise process is a
 stationary, zero mean, colored Gaussian process, the optimal input is also a stationary Gaussian process \cite{cover}. It is also reasonable to assume that data symbols from the two users $\bm{a}_1$ and $\bm{a}_2$ are independent. Then, the covariance matrix of each user is $\mathbb{E}[\bm{a}_k\bm{a}_k^\dagger]=\bm{R}_k, k=1,2$, and the covariance matrix $\bm{\Sigma_y}$ can be written as 
\begin{equation}
	\bm{\Sigma_y} = \tilde{\bm{G}}\left[\begin{matrix}
		\bm{R}_1 & \bm{0} \\ \bm{0} & \bm{R}_2
	\end{matrix}\right]\tilde{\bm{G}}^\dagger + \sigma_0^2\tilde{\bm{G}} \triangleq \tilde{\bm{G}}\tilde{\bm{R}}\tilde{\bm{G}}^\dagger + \sigma_0^2\tilde{\bm{G}}, \label{eq:GR}
\end{equation}
where $\bm{0}$ is an all-zero matrix of size $N\times N$. Then, mutual information expressions for the single-user rate constraints in \eqref{eqn:capregiondef} can be calculated as \begin{align}
&I_N(\bm{a}_1;\bm{y}|\bm{a}_2)= h(\bm{y}_1|\bm{a}_2)-h(\bm{y}_1|\bm{a}_1,\bm{a}_2)\\
&\leq\frac{1}{2N}\log_2\det\left(\mathbb{E}\left[(\bm{G}\bm{a}_1+\bm{\eta}_1)(\bm{G}\bm{a}_1+\bm{\eta}_1)^\dagger\right]\right) \notag\\
&\quad\quad\quad\quad\quad\quad\quad\quad\quad\quad\quad\quad-\frac{1}{2N}\log_2\det\left(\mathbb{E}\left[\bm{\eta}_1\bm{\eta}_1^\dagger\right]\right) \\
&=\frac{1}{2N}\log_2\det\left(\bm{G}\bm{R}_1\bm{G}+\sigma_0^2\bm{G}\right) -\frac{1}{2N}\log_2\det\left(\sigma_0^2\bm{G}\right) \label{eqn:numstable}\\
&=\frac{1}{2N}\log_2\det\left(\bm{I}_{N}+\sigma_0^{-2}\bm{G}\bm{R}_1\right), \label{eqn:R1N}
\end{align} 
and
\begin{equation}
I_N(\bm{a}_2;\bm{y}|\bm{a}_1)=\frac{1}{2N}\log_2\det\left(\bm{I}_{N}+\sigma_0^{-2}\bm{G}\bm{R}_2\right).\label{eqn:R2N}
\end{equation}

\begin{remark}
In order to calculate \eqref{eqn:numstable}, we need the matrix $\bm{G}$ to be invertible. Theoretically, a matrix is invertible as long as it is positive definite. However, this inversion may not be numerically stable. For root raised cosine pulses $p(t)$, numerical stability is achieved if $\delta(1+\beta)\geq1$, where $\beta$ is the roll-off factor \cite{ourpaper}. \label{rem;rem1}
\end{remark}

 Note that the matrices $\bm{G}$, $\bm{G}_{12}$, $\bm{G}_{21}$, $\bm{R}_1$ and $\bm{R}_2$ are all Toeplitz matrices. In addition, comparing \eqref{eqn:genefuncdef} and \eqref{eqn:dataspect}, we observe that
$S_{ak}(-\lambda)$ is the generating function of the matrix $\bm{R}_k$. Since $G_\delta(\lambda)$ in \eqref{eqn:folded} is an even function, $S_{ak}(-\lambda)=S_{ak}(\lambda)$. Then, applying Szeg\"o's theorem \cite{gray} and \cite[Theorem 2]{jesus} on the single-user rate constraints of \eqref{eqn:capregiondef}, we have  
\begin{equation}
  R_{ k}\leq\frac{1}{2}\int_{-\frac{1}{2}}^{\frac{1}{2}}\log_2(1+\sigma_0^{-2}S_{a{ k}}(\lambda)G_\delta(\lambda))d\lambda,~ {k}=1,2.
\end{equation}

 To find the sum-rate constraint, we first observe that $\tilde{\bm{G}}$ and $\tilde{\bm{R}}$ in \eqref{eq:GR} are block Toeplitz matrices \cite{jesus},{ \cite{kimbc}}. Then, we derive the sum-rate constraint in \eqref{eqn:capregiondef} as
 \begin{align}
&R_1+R_2 \notag\\
&\leq\underset{N\rightarrow\infty}{\lim}\frac{1}{2N}\bigg(\log_2\det\left(\mathbb{E}\left[\bm{y}\bm{y}^\dagger\right]\right) -\log_2\det\left(\mathbb{E}\left[\bm{\eta}_1\bm{\eta}_1^\dagger\right]\right) \bigg)\\
&=\underset{N\rightarrow\infty}{\lim}\frac{1}{2N}\log_2\det\left(\bm{I}_{2N}+\sigma_0^{-2}\tilde{\bm{G}}\Tilde{\bm{R}}\right). \label{eqn:sumrateconst}
 \end{align}
In \eqref{eqn:sumrateconst}, $\Tilde{\bm{G}}\Tilde{\bm{R}}$ is a block Toeplitz matrix, because the product of block Toeplitz matrices is also block Toeplitz \cite[Theorem 2]{jesus}. Then, applying \cite[Theorem 6]{jesus} on the sum-rate constraint \eqref{eqn:sumrateconst} we write
 \begin{align}
   &\underset{N\rightarrow\infty}{\lim}I_N(\bm{a}_1,\bm{a}_2;\bm{y}) \notag\\
   &=\frac{1}{2}\int_{-\frac{1}{2}}^{\frac{1}{2}}\log_2\sigma_0^{-2}\left|\begin{matrix}1+S_{a1}(\lambda)G_\delta(\lambda) & S_{a2}(\lambda)G_{12,\delta}(-\lambda) \\ S_{a1}(\lambda)G_{21,\delta}(-\lambda) & 1+S_{a2}(\lambda)G_\delta(\lambda) \end{matrix}\right|d\lambda \\
   &=\frac{1}{2}\int_{-\frac{1}{2}}^{\frac{1}{2}}\log_2\bigg(1+\sigma_0^{-2}S_{a1}(\lambda)G_\delta(\lambda)+\sigma_0^{-2}S_{a2}(\lambda)G_\delta(\lambda) \notag\\
   &\qquad\qquad+\sigma_0^{-4}S_{a1}(\lambda)S_{a2}(\lambda)\left[|G_\delta(\lambda)|^2-\left|G_{12,\delta}(\lambda)\right|^2\right]\bigg) d\lambda,
 \end{align}
where $G_{12,\delta}(\lambda)$ is the generating function of the matrix $\bm{G}_{12}$ obtained via \eqref{eqn:genefuncdef} and written as
\begin{align}
  G_{12,\delta}(\lambda)&=\sum_{n=\infty}^{\infty}g(n\delta T + (\tau_1-\tau_2))e^{j2\pi\lambda n} \\
  &=\frac{1}{\delta T}\sum_{n=-\infty}^{\infty}G\left(\frac{\lambda-n}{\delta T}\right)e^{j2\pi(\tau_1-\tau_2)\frac{\lambda-n}{\delta T}}.
\end{align} Similarly $G_{21,\delta}(\lambda)$ is the generating function of the matrix $\bm{G}_{21}$. It is easy to see that $G_{12,\delta}(\lambda)=\left(G_{21,\delta}(\lambda)\right)^*$ and $|G_{12,\delta}(\lambda)|^2=|G_{12,\delta}(\lambda)|^2$.
\begin{theorem}
  The capacity region of the two-user asynchronous MAC with FTN is given as
\begin{align}
		&C= \underset{{\int_{-\frac{1}{2}}^{\frac{1}{2}} S_{k}(\lambda)d\lambda\leq P_i,~S_{k}(\lambda)\geq0,~ \forall\lambda\in[-\frac{1}{2},\frac{1}{2}], k=1,2}}{\bigcup}\bigg\{(R_1, R_2):\notag \\
 & R_1\leq\frac{1}{2}\int_{-\frac{1}{2}}^{\frac{1}{2}}\log_2(1+\sigma_0^{-2}S_{1}(\lambda))d\lambda \notag \\
		  & R_2\leq\frac{1}{2}\int_{-\frac{1}{2}}^{\frac{1}{2}}\log_2(1+\sigma_0^{-2}S_{2}(\lambda))d\lambda \label{eqn:capregspec}
   \end{align} \begin{align}
		 & R_1+R_2\leq\frac{1}{2}\int_{-\frac{1}{2}}^{\frac{1}{2}}\log_2\bigg(1+\sigma_0^{-2}S_{1}(\lambda)+\sigma_0^{-2}S_{2}(\lambda) \notag \\
  &\quad\quad\quad\quad\quad\quad +\sigma_0^{-4}S_{1}(\lambda)S_{2}(\lambda)\left[1-\left|\frac{G_{12,\delta}(\lambda)}{G_\delta(\lambda)}\right|^2\right]\bigg) d\lambda \bigg\}. \notag 
\end{align} 
\end{theorem}
\begin{remark} When $\delta=1$, \eqref{eqn:capregspec} reduces to \cite[Theorem 1]{ganji} for orthogonal signaling.
\end{remark}
{  
\begin{remark}
  When $\delta=1/(1+\beta)$, $|G_{12,\delta}(\lambda)|^2=|G_{\delta}(\lambda)|^2$ and \eqref{eqn:capregspec} reduces to the capacity region of synchronous MAC with FTN, regardless of the time difference. In other words, with fast enough FTN, asynchronous transmission loses its meaning as there is no spectral aliasing in $G_{12,\delta}(\lambda)$. For the optimal power allocation, users perform spectrum shaping for point-to-point FTN as discussed in \cite{ourpaper}. 
  \label{remark3} 
  \end{remark}
  \begin{remark}
When $\delta>1/(1+\beta)$, as in Remark~\ref{remark3}, the input distribution achieving the individual rate upper bound in \eqref{eqn:capregspec} has a covariance matrix $\bm{G}^{-1}$ scaled according to the power constraint \cite{ourpaper}. However, this same distribution does not achieve the sum-rate upper bound. Therefore, we conduct power optimization similar to \cite{verdu} to obtain the capacity region, and find that it is smooth, and there are no sharp corners as in synchronous MAC \cite{cover}. \label{remark4}
\end{remark}}
\vspace{-0.4cm}
\section{An Alternative Capacity Calculation}
The capacity region of the asynchronous MAC with finite memory is defined as \cite{verdu} 
\begin{equation}
C=\text{closure}\left(\underset{N\rightarrow\infty}{\liminf}~C_N\right). \label{eqn:detregdef}
\end{equation} Here $C_N$ is the achievable region for $N$ symbols defined as 
\begin{equation}
	\begin{aligned}
		C_N= \underset{p(\bm{a}_1)p(\bm{a}_2)}{\bigcup}\bigg\{(R_1, R_2): & 0\leq R_1\leq I(\bm{a}_1;\bm{y}| \bm{a}_2) \\
		 & 0\leq R_2\leq I(\bm{a}_2;\bm{y}| \bm{a}_1) \\
		 & 0\leq R_1+R_2\leq I(\bm{a}_1, \bm{a}_2;\bm{y}) \bigg\},
	\end{aligned}\label{eqn:CNmutual}
\end{equation}
 where $p(\bm{a}_1)$ and $p(\bm{a}_2)$ mean the distribution of $\bm{a}_1$ and $\bm{a}_2$ respectively. Capacity for an arbitrary MAC with infinite memory cannot be defined in general. However, we will show that this same expression is valid as long as the limit exists \cite{cover,brandenburg}. Therefore, in this section we calculate \eqref{eqn:detregdef} and prove that it is equal to the capacity region in \eqref{eqn:capregspec}. 

 By combining \eqref{eqn:R1N}, \eqref{eqn:R2N} and \eqref{eqn:sumrateconst}, $C_N$ for the asynchronous MAC with FTN can be written as 
 \begin{equation}
	\begin{aligned}
		C_N= \underset{\substack{\frac{1}{N\delta T}\text{tr}(\bm{GR}_k)\leq P_k \\ \bm{R}_k\succeq0, k=1,2 }}{\bigcup}\bigg\{(R_1, R_2):  R_1\leq\frac{1}{2N}\log_2\left|\bm{I}_{N}+\sigma_0^{-2}\bm{G}\bm{R}_1\right| \\
		  R_2\leq\frac{1}{2N}\log_2\left|\bm{I}_{N}+\sigma_0^{-2}\bm{G}\bm{R}_2\right| \\
		  R_1+R_2\leq \frac{1}{2N}\log_2\left|\bm{I}_{2N}+\sigma_0^{-2}\tilde{\bm{G}}\tilde{\bm{R}}\right| \bigg\}.
	\end{aligned}\label{eqn:CNdet}
\end{equation}

In order to further push the sum-rate upper bound, we suggest the novel derivation in 
 \eqref{eqn:manipulation0}-\eqref{eqn:manipulation5}. In order for \eqref{eqn:manipulation0} to be computable, we need Remark \ref{rem;rem1} to be valid. 
In step (a), we define $\bm{\Phi}\triangleq\bm{G}^{-\frac{1}{2}}\bm{G}_{12}\bm{G}^{-\frac{1}{2}}$, $\bm{\Psi}_1\triangleq\bm{G}^{\frac{1}{2}}\bm{R}_{1}\bm{G}^{\frac{1}{2}}$ and $\bm{\Psi}_2\triangleq\bm{G}^{\frac{1}{2}}\bm{R}_{2}\bm{G}^{\frac{1}{2}}$, where $\bm{\Psi}_1$ and $\bm{\Psi}_2$ are Hermitian matrices. In (b), we perform singular value decomposition on the matrix $\bm{\Phi}=\bm{U_\Phi\Lambda_\Phi V_\Phi}^\dagger$, where $\bm{\Lambda_\Phi}=\text{diag}\{\lambda_1, \lambda_2, \dots, \lambda_N\}$ is a diagonal matrix and $\lambda_i, i=1,\dots,N$ are the singular values of $\bm{\Phi}$. In (c) we define $\tilde{\bm{\Psi}}_1\triangleq\bm{U_\Phi}^\dagger\bm{\Psi}_1\bm{U_\Phi}$ and $\tilde{\bm{\Psi}}_2\triangleq\bm{V_\Phi}^\dagger\bm{\Psi}_2\bm{V_\Phi}$, where $\psi_{1i}$ and $\psi_{2i}$ are the diagonal entries of $\tilde{\bm{\Psi}}_1$ and $\tilde{\bm{\Psi}}_2$.
\begin{figure*}[t]
\begin{align}
		I(\bm{a}_1,\bm{a}_2;\bm{y})
		&=\frac{1}{2N}\log_2\det\left(\bm{I}+\sigma_0^{-2}\left[\begin{matrix}
			\bm{G}^{\frac{1}{2}} & \bm{0} \\ \bm{0} & \bm{G}^{\frac{1}{2}}
		\end{matrix}\right]
		\left[\begin{matrix}
			\bm{I} & \bm{G}^{-\frac{1}{2}}\bm{G}_{12}\bm{G}^{-\frac{1}{2}} \\ \bm{G}^{-\frac{1}{2}}\bm{G}_{21}\bm{G}^{-\frac{1}{2}} & \bm{I}
		\end{matrix}\right]
		\left[\begin{matrix}
			\bm{G}^{\frac{1}{2}} & \bm{0} \\ \bm{0} & \bm{G}^{\frac{1}{2}}
		\end{matrix}\right]
	\left[\begin{matrix}
			\bm{R}_{1} & \bm{0} \\ \bm{0} & \bm{R}_{2} 
		\end{matrix}\right]\right) \label{eqn:manipulation0}\\
		&=\frac{1}{2N}\log_2\det\left(\bm{I}+\sigma_0^{-2}\left[\begin{matrix}
			\bm{I} & \bm{G}^{-\frac{1}{2}}\bm{G}_{12}\bm{G}^{-\frac{1}{2}} \\ \bm{G}^{-\frac{1}{2}}\bm{G}_{21}\bm{G}^{-\frac{1}{2}} & \bm{I}
		\end{matrix}\right]
	\left[\begin{matrix}
			\bm{G}^{\frac{1}{2}}\bm{R}_{1}\bm{G}^{\frac{1}{2}} & \bm{0} \\ \bm{0} & \bm{G}^{\frac{1}{2}}\bm{R}_{2}\bm{G}^{\frac{1}{2}} 
		\end{matrix}\right]\right) \label{eqn:manipulation1}\\
	&\overset{(a)}{=}\frac{1}{2N}\log_2\det\left(\bm{I}+\sigma_0^{-2}\left[\begin{matrix}
			\bm{I} & \bm{\Phi} \\ \bm{\Phi}^\dagger & \bm{I}
		\end{matrix}\right]
	\left[\begin{matrix}
			\bm{\Psi}_1 & \bm{0} \\ \bm{0} & \bm{\Psi}_2 
		\end{matrix}\right]\right) \label{eqn:manipulation2}\\
	&\overset{(b)}{=} \frac{1}{2N}\log_2\det\left(\bm{I}+\sigma_0^{-2}\left[\begin{matrix}
		\bm{U_\Phi} & \bm{0} \\ \bm{0} & \bm{V_\Phi}
	\end{matrix}\right]
	\left[\begin{matrix}
		\bm{I} & \bm{\Lambda_\Phi} \\ \bm{\Lambda_\Phi}^\dagger & \bm{I}
	\end{matrix}\right]
\left[\begin{matrix}
	\bm{U_\Phi}^\dagger & \bm{0} \\ \bm{0} & \bm{V_\Phi}^\dagger
\end{matrix}\right]
	\left[\begin{matrix}
		\bm{\Psi}_1 & \bm{0} \\ \bm{0} & \bm{\Psi}_2 
	\end{matrix}\right]\right) \label{eqn:manipulation3} \\
		&=\frac{1}{2N}\log_2\det\left(\bm{I}+\sigma_0^{-2}
		\left[\begin{matrix}
			\bm{I} & \bm{\Lambda_\Phi} \\ \bm{\Lambda_\Phi}^\dagger & \bm{I}
		\end{matrix}\right]
\left[\begin{matrix}
	\bm{U_\Phi}^\dagger\bm{\Psi}_1\bm{U_\Phi} & \bm{0} \\ \bm{0} & \bm{V_\Phi}^\dagger\bm{\Psi}_2\bm{V_\Phi} 
\end{matrix}\right]\right) \label{eqn:manipulation4} \\
	&\overset{(c)}{=}\frac{1}{2N}\log_2\det\left(\bm{I}+\sigma_0^{-2}
\left[\begin{matrix}
	\bm{I} & \bm{\Lambda_\Phi} \\ \bm{\Lambda_\Phi}^\dagger & \bm{I}
\end{matrix}\right]
\left[\begin{matrix}
	\tilde{\bm{\Psi}}_1 & \bm{0} \\ \bm{0} & \tilde{\bm{\Psi}}_2
\end{matrix}\right]\right) \label{eqn:manipulation5} .
\end{align}
\end{figure*}
Then, we apply \cite[Lemma 2]{verdu} on \eqref{eqn:manipulation5} to upper bound the mutual information as
\begin{align}
	&I(\bm{a}_1,\bm{a}_2;\bm{y})=\frac{1}{2N}\log_2\left(\bm{I}+\sigma_0^{-2}
	\left[\begin{matrix}
		\bm{I} & \bm{\Lambda_\Phi} \\ \bm{\Lambda_\Phi}^\dagger & \bm{I}
	\end{matrix}\right]
	\left[\begin{matrix}
		\tilde{\bm{\Psi}}_1 & \bm{0} \\ \bm{0} & \tilde{\bm{\Psi}}_2
	\end{matrix}\right]\right) \notag \\
	&\leq \frac{1}{2N}\sum_{i=0}^{N-1}\log_2\left(1+\frac{\psi_{1i}}{\sigma_0^2}+\frac{\psi_{2i}}{\sigma_0^2}+\frac{\psi_{1i}\psi_{2i}}{\sigma_0^4}(1-|\lambda_i|^2)\right). \label{eqn:infoub}
\end{align} 
 The equality in \eqref{eqn:infoub} is achieved when $\tilde{\bm{\Psi}}_1$ and $\tilde{\bm{\Psi}}_2$ are diagonal matrices. Moreover, in order for \cite[Lemma 2]{verdu} to be valid or the upper bound to be achieved, we need the complex scalars $\lambda_i$ to satisfy $|\lambda_i|\leq1, i=0,1,...,N-1$. 
 Let $\bm{v}=[v[0],v[1],\dots,v[2N-1]]^T$ be a non-zero vector, where $[v[0],...,v[N-1]]=[a_1[0],...,a_1[N-1]]$ and $[v[N],...,v[2N-1]] =[a_2[1],...,a_2[N]]$. Then, the quadratic form $\bm{v}^\dagger\tilde{\bm{G}}\bm{v}$ is the energy of signal $x_1(t)+x_2(t)$, \cite{verdu}. Therefore, as long as $\bm{v}$ is not zero, $\bm{v}^\dagger\tilde{\bm{G}}\bm{v}$ will be greater than zero. Thus, we conclude that the matrix $\tilde{\bm{G}}$ is positive definite.
 Then, we have \useshortskip
\begin{align}
	\bm{v}^\dagger\tilde{\bm{G}}\bm{v}&\overset{(a)}{=}\tilde{\bm{v}}^\dagger\left[\begin{matrix}
		\bm{G}^{\frac{1}{2}} & \bm{0} \\ \bm{0} & \bm{G}^{\frac{1}{2}}
	\end{matrix}\right]^\dagger\tilde{\bm{G}}\left[\begin{matrix}
	\bm{G}^{\frac{1}{2}} & \bm{0} \\ \bm{0} & \bm{G}^{\frac{1}{2}}
\end{matrix}\right]\tilde{\bm{v}} \notag\\
&=\tilde{\bm{v}}^\dagger\left[\begin{matrix}
	\bm{I} & \bm{\Phi} \notag\\ \bm{\Phi}^\dagger & \bm{I}
\end{matrix}\right]\tilde{\bm{v}} >0, \notag
\end{align} where (a) is because $\bm{v}\triangleq\left[\begin{matrix}
\bm{G}^{\frac{1}{2}} & \bm{0} \\ \bm{0} & \bm{G}^{\frac{1}{2}}
\end{matrix}\right]\tilde{\bm{v}}$. Thus, the matrix $\left[\begin{matrix}
\bm{I} & \bm{\Phi} \notag\\ \bm{\Phi}^\dagger & \bm{I}
\end{matrix}\right]$ is positive definite as well. Then, according to \cite{matrices}, $|\lambda_i|<1, i=0,1,...,N-1$. 

Next, the upper bound for $I(\bm{a}_1;\bm{y}|\bm{a}_2)$ and $I(\bm{a}_2;\bm{y}|\bm{a}_1)$ can be obtained as in \cite{verdu} as
\begin{align}
	I(\bm{a}_1;\bm{y}|\bm{a}_2)&\leq \frac{1}{2N}\sum_{i=0}^{N-1}\log_2\big(1+\frac{\psi_{1i}}{\sigma_0^2}\big) \label{eqn:detsingleub} \\
	I(\bm{a}_2;\bm{y}|\bm{a}_1)&\leq \frac{1}{2N}\sum_{i=0}^{N-1}\log_2\big(1+\frac{\psi_{2i}}{\sigma_0^2}\big). \label{eqn:detsingleub2}
\end{align}
The power constraint for each user in this $N$-block asynchronous multiple access channel with FTN is calculated as $\frac{1}{N\delta T}\text{tr}(\bm{G}\bm{R}_k)$, $k=1,2$, where \useshortskip
\begin{align}
	\text{tr}(\bm{G}\bm{R}_k)&=\text{tr}(\bm{G}^{\frac{1}{2}}\bm{R}_k\bm{G}^{\frac{1}{2}})=\sum_{i=0}^{N-1}\psi_{ki} \leq N \delta T P_k , \label{eqn:detconstraint}
\end{align} and $\psi_{ki} \geq 0$, $i = 1,\ldots, N$. Therefore, the region $C_N$ in \eqref{eqn:CNdet} is obtained using \eqref{eqn:infoub}-\eqref{eqn:detconstraint}.

\begin{lemma}
  If we have $\sum_{n=-\infty}^{\infty}|n|t_n<\infty$, then the discrete Fourier transform (DFT) vectors are asymptotically the eigenvectors of Toeplitz matrix $\bm{T}_N$.
\label{lem:sqrtGasytoep}
\end{lemma}
\begin{proof}
  Although this result is discussed in \cite{therrien}, it is not proved. 
  We refer the readers to \cite{zhang2023capacity} for the detailed proof of this lemma.
\end{proof}
\begin{corollary}
  The region defined in \eqref{eqn:detregdef} with $C_N$ defined using \eqref{eqn:infoub}, \eqref{eqn:detsingleub}, and \eqref{eqn:detsingleub2} with the power constraint in \eqref{eqn:detconstraint} is the same as the capacity region in \eqref{eqn:capregspec}.
\end{corollary}
\begin{proof}
Since the raised cosine filter $g[n]$ satisfies $\sum_{n=-\infty}^{\infty}|n|g[n]<\infty$, by Lemma \ref{lem:sqrtGasytoep}, $\bm{G}^{-\frac{1}{2}}$ is an asymptotically Toeplitz (AT) matrix. Its generating function is $G_\delta^{\frac{1}{2}}(\lambda)$. We know that $|\lambda_i|^2$'s are the eigenvalues of the Hermitian matrix $\bm{\Phi}^\dagger\bm{\Phi}$. As $\bm{G}^{-\frac{1}{2}}$ is AT and the product of AT matrices is also AT \cite[Theorem 5.3]{gray}, the product $\bm{\Phi}^\dagger\bm{\Phi}=\bm{G}^{-\frac{1}{2}}\bm{G}_{12}^\dagger\bm{G}^{-1}\bm{G}_{12}\bm{G}^{-\frac{1}{2}}$ is also AT. The generating function of $\bm{\Phi}^\dagger\bm{\Phi}$ is $\left|\frac{G_{12,\delta}(\lambda)}{G_\delta(\lambda)}\right|^2$, which is the product of the generating functions of the individual matrices in the above $\bm{\Phi}^\dagger\bm{\Phi}$ expansion.
The eigenvalues of a Toeplitz matrix asymptotically approximate the samples of its generating function \cite{kimproperty}. Thus we have $|\lambda_i|^2\approx\left|\frac{G_{12,\delta}(i/N)}{G_\delta(i/N)}\right|^2, i=0,1,\dots,N-1$. Moreover, the values $|\lambda_i|^2$ are the samples from the constant spectrum $\left|\frac{G_{12,\delta}(\lambda)}{G_\delta(\lambda)}\right|^2$. Hence the discussion in \cite{verdu} about time and frequency domain capacity region comparison applies, and we conclude that the two regions are the same. 
 \end{proof} 

\begin{figure}[t]
	\includegraphics[scale=0.368]{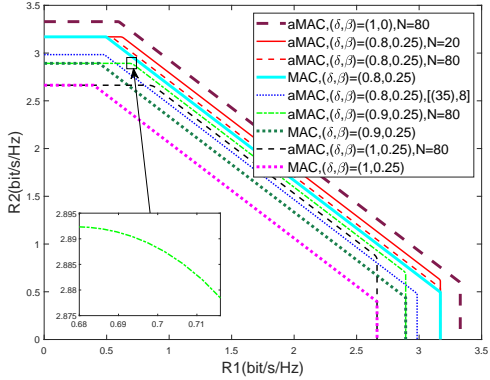}
	\caption{Capacity regions for asynchronous MAC with different $(\delta, \beta)$ pairs, { the bound in \cite[(35)]{outperform} for iid inputs}, and synchronous MAC.}
	\label{fig:fourinone}
\end{figure}

\vspace{-2mm}
\section{Numerical Results}
In this section, we plot the capacity region { \eqref{eqn:infoub}-\eqref{eqn:detconstraint}.} for root raised cosine pulses $p(t)$. We set the SNR ($P_k/\sigma_0^2$) for both users to be $20$ dB, and the signaling period $T=1$. 

{ In Fig. \ref{fig:fourinone}, in the legend, aMAC and MAC respectively mean asynchronous and synchronous transmission. When $\delta = 1$, there is Nyquist transmission, and if $\delta<1$, FTN is utilized. The $N$ values used to obtain the curves are also indicated in the legend. The curves, for which there is no $N$, can be plotted independent of $N$; i.e., for $N$ approaches infinity. For all the asynchronous simulations in this figure, except \cite[(35)]{outperform}, we set the time difference $\tau=\frac{\delta T}{2}$. In the figure, we plot MAC, $(\delta,\beta)= (1,0.25)$ as a baseline, and aMAC, $(\delta,\beta)= (1,0)$ as an upper bound. Next, we present results to reveal the gains due to optimal power allocation, FTN, and asynchronism.

In Fig.~\ref{fig:fourinone}, first we confirm that as $N$ increases from 20 to 80, the aMAC, $(\delta,\beta)= (0.8,0.25)$ curve converges to the MAC, $(\delta,\beta)= (0.8,0.25)$ curve as explained in Remark~\ref{remark3}. Furthermore, as explained in Remark~\ref{remark4} and shown in the zoomed in section, we have a smooth corner for aMAC, $(\delta,\beta)=(0.9,0.25)$. %
When we compare MAC, $(\delta,\beta)= (0.8,0.25)$ with aMAC, $(\delta,\beta)= (0.8,0.25)$, {\cite[(35)]{outperform}}, we observe that optimal power allocation in aMAC with FTN improves the rate region, both in terms of individual rates and also in the sum rate. %
In Fig.~\ref{fig:fourinone} when we compare aMAC, $(\delta,\beta)=(0.8,0.25)$ with aMAC, $(\delta,\beta)=(1,0.25)$; and MAC, $(\delta,\beta)=(0.8,0.25)$ with MAC $(\delta,\beta)=(1,0.25)$, we observe the individual gain due to FTN respectively in aMAC and MAC and conclude that the whole rate region enlarges. %
Similarly, when we compare aMAC, $(\delta,\beta)=(0.9,0.25)$ with MAC, $(\delta,\beta)=(0.9,0.25)$; and aMAC, $(\delta,\beta)=(1,0.25)$ with MAC, $(\delta,\beta)=(1,0.25)$, we observe the individual gain due to asynchronism respectively for FTN and for Nyquist transmission. Unlike FTN, asynchronism can only improve the sum rate, but not individual rates. %
In Fig.~\ref{fig:fourinone}, we also look into the effect of different $\delta$ and compare aMAC with optimal power allocation for $\delta= 0.8,0.9$ and $1$. Confirming the results about $\delta$ in point-to-point communications \cite{ourpaper}, we see that for a given $\beta$ value, $\delta(1+\beta)$ must be as close to $1$ as possible for better performance. }

Finally, we study the influence of time difference $\tau$ in Fig. 2. We can see that when the time difference between two users is half the sampling period $\frac{\delta T}{2}$, the performance is better than the performance with other values of $\tau$. This suggests that \cite[Proposition 2]{ganji} also holds in the presence of FTN.

\begin{figure}[t]
	\centering
	\includegraphics[scale=0.49]{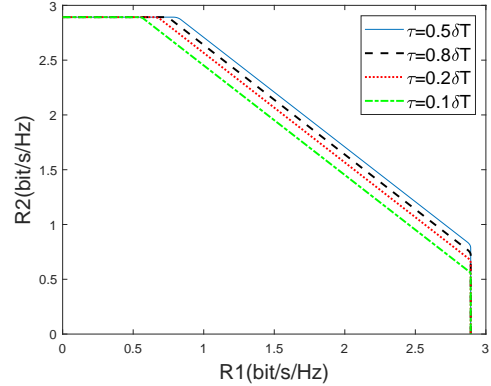}
	\caption{Capacity regions for asynchronous MAC with FTN with different time differences $\tau$ between users, { $(\delta, \beta)=(0.9,0.25)$, $N=20$}.}
	\label{fig:difftaubeta}
\end{figure}

\vspace{-0.2cm}
\section{Conclusion}

In this paper we derive the capacity region of asynchronous multiple access channels with FTN both in frequency and time domains. We find that optimal power allocation is necessary to obtain the capacity region. We also show that the capacity region definition for finite memory MAC can be generalized to infinite memory MAC. As a side result, we prove that the DFT vectors are asymptotically the eigenvectors of the Toeplitz matrix $\bm{T}_N$ as long as $\sum_{n=-\infty}^{\infty}|n|t_n<\infty$.
{ We leave the extension to more than two users for future work.}






\vspace{-0.3cm}
\bibliographystyle{IEEEtran}
\bibliography{main}




\end{document}